\documentclass[a4paper,USenglish,cleveref]{lipics-v2021}

\usepackage[utf8]{inputenc}
\usepackage{amsmath,amssymb,amsfonts}
\usepackage{latexsym}
\usepackage{xspace}
\usepackage{verbatim}
\usepackage{mathtools}
\usepackage{xcolor}
\definecolor{defblue}{rgb}{0.121,0.47,0.705}
\definecolor{linkblue}{rgb}{0.082,0.310,0.537}

\definecolor{dark brown}{rgb}{0.651, 0.337, 0.157}

\hypersetup{colorlinks=true,
    linkcolor=linkblue,
    anchorcolor=linkblue,
    citecolor=linkblue,
    filecolor=linkblue,
    menucolor=linkblue,
    urlcolor=linkblue,
    bookmarksopen=true,
    bookmarksopenlevel=2,
    bookmarksnumbered=true,
    plainpages=false,
}

\nolinenumbers
\hideLIPIcs
    
\let\emph\relax
\DeclareTextFontCommand{\emph}{\color{defblue}\em}
\usepackage{graphicx}
\graphicspath{{figures/}}
\usepackage{url}

\Crefname{figure}{Figure}{Figures}
\crefname{proposition}{Proposition}{Propositions}
\Crefname{observation}{Observation}{Observations} %
\crefname{observation}{Observation}{Observations} %

\DeclarePairedDelimiter\set{\{}{\}}
\DeclarePairedDelimiter\abs{\lvert}{\rvert}

\DeclarePairedDelimiter\croc{\langle}{\rangle}
\def\Oh{\ensuremath{\mathcal{O}}}
\DeclareMathOperator{\Obs}{obs}

\title{Outside-Obstacle Representations\\ with All Vertices on the Outer Face}
\titlerunning{Outside-Obstacle Representations with All Vertices on the Outer Face}

\author{Oksana Firman}{Institut für Informatik, Universität Würzburg,
  Germany}{firstname.lastname@uni-wuerzburg.de}%
{https://orcid.org/0000-0002-9450-7640}{Partially
  funded by DFG project Wo~758/9-1}
\author{Philipp Kindermann}{Informatikwissenschaften, Universität Trier, Germany}{kindermann@uni-trier.de}{https://orcid.org/0000-0001-5764-7719}{}
\author{Jonathan Klawitter}{School of Computer Science, University of Auckland, New Zealand}{jo.klawitter@gmail.com}{https://orcid.org/0000-0001-8917-5269}{Beyond Prediction Data Science Research Programme (MBIE grant UOAX1932).}
\author{Boris Klemz}{Institut für Informatik, Universität Würzburg, Germany}{firstname.lastname@uni-wuerzburg.de}{https://orcid.org/0000-0002-4532-3765}{}
\author{Felix Klesen}{Institut für Informatik, Universität Würzburg, Germany}{firstname.lastname@uni-wuerzburg.de}{https://orcid.org/0000-0003-1136-5673}{}
\author{Alexander Wolff}{Institut für Informatik, Universität Würzburg, Germany \and \url{https://www.informatik.uni-wuerzburg.de/en/algo/team/wolff-alexander}}{}{https://orcid.org/0000-0001-5872-718X}{}

\authorrunning{Firman, Kindermann, Klawitter, Klemz, Klesen, and Wolff} 

\Copyright{Oksana Firman, Philipp~Kindermann, Jonathan~Klawitter, Boris~Klemz, Felix~Klesen, and Alexander~Wolff}
\relatedversion{A preliminary version of this work has appeared in
  Proc.\ 30th International Symposium Graph Drawing and Network
  Visualization (GD’22), Springer-Verlag, 2023~\cite{fkkkkw-ooravof-GD22}.}


\begin{document}
\maketitle

\pdfbookmark[1]{Abstract}{Abstract} 
\begin{abstract}
  An \emph{obstacle representation} of a graph~$G$ consists of a set
  of polygonal obstacles and a drawing of~$G$ as a \emph{visibility graph} 
  with respect to the obstacles: 
  vertices are mapped to points and edges to straight-line segments 
  such that each edge avoids all obstacles whereas each non-edge 
  intersects at least one obstacle.
  Obstacle representations have been investigated quite
  intensely over the last few years.
  Here we focus on \emph{outside-obstacle representations} (OORs) that use
  only one obstacle in the outer face of the drawing.
  It is known that every outerplanar graph admits such a representation. 
  
  We strengthen this result by showing that every (partial) 2-tree has
  an OOR.  We also consider restricted versions of OORs
  where the vertices of the graph form a convex or even a regular polygon.
  We characterize when the complement of a tree and when a complete graph
  minus a simple cycle admits a convex OOR.
  We construct regular OORs for all (partial) outerpaths,
  cactus graphs, and grids.
\end{abstract}

\begin{figure}[tb]
  \begin{subfigure}[t]{.39\linewidth}
	\centering
	\includegraphics{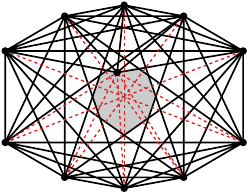}
        \caption{A graph that admits an inside-obstacle representation but
          no outside-obstacle representation~\cite{clpw-ov1o-GD16}.}
        \label{fig:11-vertex-graph}
  \end{subfigure}
  \hfill
  \begin{subfigure}[t]{.58\linewidth}
	\centering
	\includegraphics{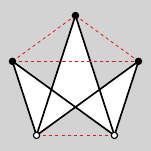}
        \caption{The complete bipartite graph $K_{2,3}$ is the smallest
          graph that contains a cycle and admits an
          outside-obstacle representation but no inside-obstacle
          representation~\cite{clpw-ov1o-GD16}.
        }
  \end{subfigure}
  \caption{Inside- and outside-obstacle representations of graphs.
    The gray regions represent the obstacles; the dashed line segments
    represent the non-edges.}
  \label{fig:teaser}
\end{figure}


\section{Introduction} 
\label{sec:intro}

Recognizing graphs that have a certain type of geometric
representation is a well-established field of research dealing with, 
e.g., geometric intersection graphs, visibility graphs, and graphs admitting certain contact representations.   
Given a set $\mathcal C$ of \emph{obstacles}
(here, connected polygonal regions, which we consider to be open)
and a set~$P$ of points in the plane, the \emph{visibility
graph}~$G_{\mathcal C}(P)$ has a vertex for each point
in~$P$ and an edge~$pq$ for any two points~$p$ and~$q$ in~$P$ that can
\emph{see} each other, that is, the line segment $\overline{pq}$
connecting $p$ and $q$ does not intersect any obstacle in~$\mathcal C$
and does not contain any other point of~$P$.
An \emph{obstacle representation} of a graph $G$ consists of a
set~$\mathcal{C}$ of obstacles in the plane and a mapping of the
vertices of~$G$ to a set~$P$ of points such that~$G = G_{\mathcal C}(P)$.
Note that, in such a representation, for each \emph{non-edge} $pq$
of~$G$, that is, for each non-adjacent vertex pair, the line segment
$\overline{pq}$ {\em must} intersect an obstacle in~$\mathcal C$.
The vertex--point mapping defines a straight-line drawing~$\Gamma$ of
$G_{\mathcal C}(P)$. We planarize~$\Gamma$ by replacing all
intersection points by dummy vertices.  The faces of this
planarization are open polygonal regions.  We call the unbounded
face the \emph{outer face} of~$\Gamma$.  We differentiate between
two types of obstacles: \emph{outside} obstacles lie in the 
outer face of the drawing, and \emph{inside} obstacles
lie in the complement of the outer face; see~\cref{fig:teaser}.

Every graph trivially admits an obstacle representation:
take an arbitrary straight-line drawing without collinear vertices and
``fill'' each face with an obstacle.
This, however, can lead to a large number of obstacles, 
which motivates the optimization problem of finding an obstacle
representation with the minimum number of obstacles. 
For a graph $G$, the \emph{obstacle number} $\Obs(G)$ is
the smallest number of obstacles that suffice to represent~$G$ as a
visibility~graph.

While most previous work (discussed below) has focused on establishing
bounds for the obstacle number of general graphs, we are interested in
understanding which graphs can be represented by a {\em single} obstacle.
In particular, we focus on \emph{outside obstacle representations (OORs)},
that is, obstacle representations with a single outside obstacle
and without any inside obstacles.  
For such a representation, it suffices to specify the positions of
the vertices of the given graph;
the outside obstacle is simply the whole outer face of the representation.
In an OOR, every non-edge of the graph must thus intersect the outer face.
We consider three special types of OORs:
In~a \emph{convex} OOR, the vertices must be in convex position;
in a \emph{circular} OOR, the vertices must lie on a circle; and	
in a \emph{regular} OOR, the vertices must form a regular $n$-gon, 
where $n$ is the number of vertices of the graph.
For examples, refer to \cref{fig:examples}.

\begin{figure}[tb]
  \centering
  \begin{subfigure}[t]{.36\linewidth}
	\centering
	\includegraphics{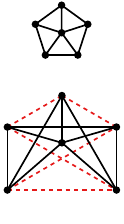}
	\caption{The wheel graph~$W_6$ does not ad\-mit a convex OOR
          (\cref{clm:ngon:6vertices}).}
	\label{fig:not-convex}
  \end{subfigure}
  \hfill
  \begin{subfigure}[t]{.42\linewidth}
	\centering
	\includegraphics[page=2]{outerplanar7}
	\caption{Graph that admits a circular but not a regular OOR
          (move vertex~7 towards $x$).}
	\label{fig:outerplanar7}
  \end{subfigure}
  \hfill
  \begin{subfigure}[t]{.15\linewidth}
	\centering
	\includegraphics{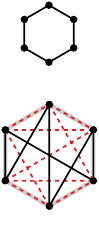}
	\caption{$C_6$ admits a regular OOR.}
	\label{fig:regularOOR}
  \end{subfigure}
  \caption{Non-convex, circular, and regular outside-obstacle
    representations (OORs).  Graph edges are solid black, non-edges
    are dashed red (in (b) only one non-edge is highlighted).}
  \label{fig:examples}
\end{figure}

In general, the class of graphs that admit an OOR is not closed under taking subgraphs, 
but the situation is different for graphs admitting a \emph{reducible} OOR, 
meaning that all of its edges are incident to the outer face.
In this case, we can simply extend the outside obstacle to intersect
any edge we want to remove.

\begin{observation} \label{clm:reducible}
  If a graph $G$ admits a reducible OOR, then every subgraph of~$G$
  also admits such a representation.
\end{observation}

\subparagraph{Previous Work.}
Alpert, Koch, and Laison~\cite{akl-ong-DCG10} introduced the notion of
the obstacle number of a graph.
They also introduced the notion of an \emph{inside obstacle representations (IOR)}, 
that is, an obstacle representation without an outside obstacle. 
They characterized the class of graphs that have an IOR with a single convex obstacle 
and showed that every outerplanar graph has an OOR.
Chaplick, Lipp, Park, and Wolff~\cite{clpw-ov1o-GD16} proved that the class
of graphs with an IOR is incomparable with the class of graphs with an
OOR by observing that no path admits an IOR and the graph in
\cref{fig:11-vertex-graph}, which admits an IOR, does not admit an OOR.
They showed that any graph with at most seven vertices has an OOR, 
which does not hold for a specific 8-vertex~graph.
They also showed that the following \emph{sandwich version} of the
outside-obstacle representation problem is NP-hard: 
Given two graphs~$G$ and~$H$ with the same vertex set~$V$ such that
$G$ is a subgraph of~$H$, is there a graph~$K$ with vertex set~$V$
that is a supergraph of~$G$ and a subgraph of~$H$ and admits an
outside-obstacle representation?  Analogous hardness results hold with
respect to inside and general obstacles.

Alpert, Koch, and Laison~\cite{akl-ong-DCG10} further showed that $\Obs(K^*_{a,b}) \le 2$ for any $a \le b$, 
where $K^*_{a,b}$ is the complete bipartite graph $K_{a,b}$ minus a matching of size~$a$.
They also proved that $\Obs(K^*_{5,7}) = 2$.
Pach and Sar\i\"{o}z~\cite{ps-sglon-GC11} showed that $\Obs(K^*_{5,5})=2$. 
Berman, Chappell, Faudree, Gimbel, Hartman, and Williams~\cite{bcfghw-gong1-JGAA16} suggested 
some necessary conditions for a graph to have obstacle number~1.
They gave a SAT formula that they used to find a {\em planar}
10-vertex graph~$X_4$ (of treewidth~4) that has no 1-obstacle
representation; see \cref{fig:x4}.

\begin{figure}[tbh]
  \begin{subfigure}[t]{.25\linewidth}
    \centering    
    \includegraphics[page=1]{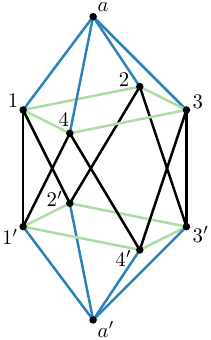}
    \caption{the graph $X_4$}
  \end{subfigure}
  \hfill
  \begin{subfigure}[t]{.71\linewidth}
    \centering    
    \includegraphics[page=2]{x4new}
    \caption{a 2-obstacle representation of
      $X_4$~\cite{bcfghw-gong1-JGAA16}}
  \end{subfigure}
  
  \caption{The so-called gyroelongated square bipyramid $X_4$, which
    has treewidth~4, does not admit a representation with a single
    obstacle.}
  \label{fig:x4}
\end{figure}

Obviously, any $n$-vertex graph has obstacle number~$\Oh(n^2)$.
Balko, Cibulka, and Valtr~\cite{bcv-dgusno-DCG18} improved this to~$\Oh(n\log n)$.
On the other hand, Balko, Chaplick, Gupta, Hoffmann, Valtr, and Wolff~\cite{bcghvw-bcong-ESA22} 
showed that there are $n$-vertex~graphs whose obstacle number
is $\Omega(n/\log\log n)$, improving previous lower bounds, e.g.,
\cite{akl-ong-DCG10,dm-on-EJC15,mpp-lbong-EJC12,mps-glon-WG10}.
They also showed that, when restricting obstacles to {\em convex} polygons,
sometimes a linear number of obstacles is needed.
Furthermore, they showed that computing the obstacle number
of a graph~$G$ is fixed-parameter tractable in the vertex cover number
of~$G$ and that it is NP-hard to decide whether a given graph
admits an obstacle representation using a given simple polygon as
(outside-) obstacle.

\subparagraph{Contribution.}
We do a detailed analysis of the class of graphs that admit OORs:
we construct graphs that do not admit convex OORs and we find
subclasses that admit convex, circular, or regular OORs.
We first strengthen the result of Alpert, Koch, and Laison~\cite{akl-ong-DCG10} 
regarding OORs of outerplanar graphs by showing 
that every (partial) 2-tree admits a reducible OOR
with all vertices on the outer face; see \cref{sec:nonconvex}.
Note that, to obtain an OOR of an outerplanar graph $G$,
we can delete edges from a \textit{reducible} OOR of a 2-tree that contains $G$ as subgraph.
Equivalently, every graph of treewidth at most two, 
which includes outerplanar and series-parallel graphs, admits such a
representation.  Note that this result and the above-mentioned
10-vertex planar graph~$X_4$ of treewidth~4 (see \cref{fig:x4})
that does not admit any
1-obstacle representation~\cite{bcfghw-gong1-JGAA16} leave open the
question whether every
(planar) graph of treewidth~3 admits an OOR; see our list of open
problems in \cref{sec:open}.

Then we establish two combinatorial conditions for
the existence of convex OORs; see \cref{sec:conditions}.
In particular, we introduce a necessary condition 
that can be used to show that a given graph does {\em not} admit a
convex OOR as,~e.g., the graph~$W_6$ in \cref{fig:not-convex}.
We apply these conditions to characterize
when the complement of a tree and when a complete graph minus a simple cycle admits a convex OOR.
We further construct {\em regular} reducible OORs for all outerpaths, grids, and cacti; see \cref{sec:ngon}.  
The result for grids strengthens an observation by Dujmovi\'{c} and Morin~\cite[Fig.~1]{dm-on-EJC15},
who showed that grids have (outside) obstacle number~1.

\subparagraph{Notation.}
For a graph $G$, let \emph{$V(G)$} be the vertex set of $G$, and let
\emph{$E(G)$} be the edge set of~$G$.  We use $n = \abs{ V(G) }$ where
the graph is clear from the context.  Given a cyclic
order $\sigma = \croc{ v_1, v_2, \dots, v_n }$ of~$V(G)$ and indices~$i$
and~$j$ with $1 \le i \neq j \le n$, we write \emph{$[v_i,v_j)$} to refer
to the subsequence $\croc{ v_i,v_{i+1},\dots,v_{j-1} }$ of $\sigma$,
where indices are interpreted modulo $n$.  Subsequences
\emph{$(v_i,v_j)$} and \emph{$[v_i,v_j]$} are defined analogously.

\section{Outside-Obstacle Representations for Partial 2-Trees}
\label{sec:nonconvex}

The graph class of \emph{$2$-trees} is recursively defined as follows:
$K_3$ is a $2$-tree. Further, any graph is a 2-tree if it is obtained
from a $2$-tree~$G$ by introducing a new vertex~$x$ and making $x$
adjacent to both endpoints of some edge $uv$ in~$G$. 
We say that~$x$ is \emph{stacked} on~$uv$ and 
call the edges~$xu$ and $xv$ the \emph{parent edges} of~$x$.

\begin{theorem} \label{clm:nonconvex}
  Every $2$-tree admits a reducible OOR with all vertices on the outer face.
\end{theorem}
\begin{proof}
  It follows readily from the definition of $2$-trees that every
  $2$-tree~$T$ can be constructed through the following iterative
  procedure, during which every vertex is marked either as active or
  inactive.  Once a vertex is inactive, it remains inactive for the
  remainder of the construction.
  \begin{enumerate}[(S1)] 
  \item\label{enum:2tree-base}
    Start with some edge, called the \emph{base} edge and mark its vertices as \emph{inactive}.
    Stack any number of vertices (but at least one vertex)
    onto the base edge and mark the new vertices as \emph{active}.
  \item\label{enum:2tree-step} 
	Pick an active vertex~$v$ and stack any number of new vertices
        (possibly none) onto each of its two parent edges.
	The new vertices are marked as active and $v$ is marked as inactive.
  \item\label{enum:2tree-repeat} 
    While there are active vertices, repeat step~(S\ref{enum:2tree-step}).
  \end{enumerate} 
Observe that step~(S\ref{enum:2tree-step}) is performed exactly once for each vertex that is not incident to the base edge.
We construct a drawing of $T$ by geometrically implementing the iterative procedure described above,
so that after every step of the algorithm the present part of the graph is realized
as a straight-line drawing satisfying the following set of invariants:
\begin{enumerate}[({I}1)] 
  \item \label{enum:vertex}%
    Each vertex~$v$ that is not incident to the base edge is
    associated with an open circular arc $C_v$ centered at~$v$ that lies
    completely in the outer face; see \cref{fig:rotation}.  Moreover,
    the parent edges of~$v$ lie below~$v$ and contain the endpoints
    of~$C_v$.
  \item \label{enum:nonedge}%
    Each non-edge intersects the circular arc of at least one of its incident vertices.
  \item \label{enum:active}%
    For each active vertex~$v$, the region $R_v$ enclosed
    by~$C_v$ and the two parent edges of~$v$ (shaded gray in
    \cref{fig:rotation}) is \emph{empty}, meaning that~$R_v$ does not
    contain any vertex of~$T$ and does not intersect any edge or
    circular arc.  (Combined with (I\ref{enum:vertex}), it follows
    that~$R_v$ lies completely in the outer face.)
  \item \label{enum:outerface}%
    Every vertex is incident to the outer face.
\end{enumerate}

Once the procedure terminates, we have indeed obtained the desired drawing:
invariants~(I\ref{enum:vertex}) and~(I\ref{enum:nonedge}) imply 
that each non-edge passes through the outer face and, 
hence, we have indeed obtained an OOR.
Moreover, invariant~(I\ref{enum:vertex}) implies 
that each non-base edge is incident to the outer face of the drawing.
The base edge will be drawn horizontally.
By the second part of invariant~(I\ref{enum:vertex}), 
all vertices not incident to the base edge are above the base edge.
Consequently, the base edge is incident to the outer face as well 
and, hence, the representation is reducible.  
Finally, by invariant~(I\ref{enum:outerface}), every vertex belongs to the outer face.

\subparagraph{Construction.}
To carry out step~(S\ref{enum:2tree-base}), we draw the base edge horizontally 
and place the stacked vertices on a common horizontal line above the base edge; see \cref{fig:G1}.
Circular arcs that satisfy the invariants are now easy to define.

\begin{figure}[tb]
  \centering
  \begin{subfigure}[t]{0.33\textwidth}
	\centering
	\includegraphics{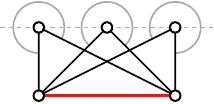}
	\caption{Step (1); the base edge is red.}
	\label{fig:G1}
  \end{subfigure}
  
  \vspace{1.5em}
  
  \begin{subfigure}[t]{0.85\textwidth}
	\centering
	\includegraphics[page=2]{rotation}
	\caption{Step (2); the hatched areas do not contain any vertices.}	
	\label{fig:rotation}
  \end{subfigure}
  \caption{Construction steps in the proof of \cref{clm:nonconvex}.}
  \label{fig:twoTree:construction}	
\end{figure}

Let~$\Gamma$ be a drawing of the graph
obtained after step~(S\ref{enum:2tree-base})
and some number of iterations of step~(S\ref{enum:2tree-step})
such that~$\Gamma$ is equipped with a set of circular arcs
satisfying the invariants~(I\ref{enum:vertex})--(I\ref{enum:outerface}).
We describe how to carry out another iteration of
step~(S\ref{enum:2tree-step}) while maintaining the invariants.

Let~$v$ be an active vertex.
By invariant~(I\ref{enum:vertex}), both parent edges of~$v$ are below~$v$.
Let~$e_\ell=v\ell$ and~$e_r=vr$ be the left and the right parent edges of~$v$,
respectively.  Let $\ell_1,\ell_2,\dots,\ell_i$ and $r_1,r_2,\dots,r_j$
be the vertices stacked onto~$e_\ell$ and~$e_r$, respectively. 
We refer to $\ell_1,\ell_2,\dots,\ell_i$ and $r_1,r_2,\dots,r_j$ as the \emph{new} vertices; the vertices of~$\Gamma$ are called \emph{old}.
See \cref{fig:rotation} for the following construction details. 
We place all the new vertices on a common horizontal line~$h$ that intersects~$R_v$ above~$v$.
The vertices $\ell_1,\ell_2,\dots,\ell_i$ are placed inside~$R_v$, to the right of the line~$\overline{e_\ell}$ extending~$e_\ell$.
Symmetrically, $r_1,r_2,\dots,r_j$ are placed inside~$R_v$, to the left of the line~$\overline{e_r}$ extending~$e_r$.
For $k \in \{1,2,\dots,i\}$, let~\emph{$W(\ell_k)$} be the smallest open
wedge with apex~$\ell_k$ that contains~$e_\ell$.  For example, in
\cref{fig:rotation}, the green shaded wedge is $W(\ell_2)$.  Note that
$W(\ell_k)$ is bounded by the two rays that go from~$\ell_k$
through~$v$ and through~$\ell$.  For $k \in \{1,2,\dots,j\}$,
let~\emph{$W(r_k)$} be defined symmetrically.

We place $\ell_1,\ell_2,\dots,\ell_i$ close enough to~$\overline{e_\ell}$ and $r_1,r_2,\dots,r_j$ close enough to~$\overline{e_r}$ 
such that the following properties are satisfied:
\begin{enumerate}[(A)] 
  \item \label{enum:propA}
    None of the parent edges of the new vertices intersects~$C_v$.
  \item \label{enum:propB}
    For each new vertex~$y$, the wedge $W(y)$ does not contain any
    vertices of~$T$.
\end{enumerate}
These properties are easy to achieve:
let~$\overline{e_\ell}(\alpha)$ be the line that is obtained by
    rotating~$\overline{e_\ell}$ clockwise around~$v$ by
    angle~$\alpha$.
Clearly, there is an angle~$\alpha^*$ such that
\begin{enumerate}[(A')] 
  \item \label{enum:propA'}
    the intersection point~$x(\alpha^*)$
    of~$\overline{e_\ell}(\alpha^*)$ and~$h$ lies in~$R_v$,
    the line segment $\overline{x(\alpha^*)\ell}$ does not
    intersect~$C_v$, and
  \item \label{enum:propB'} the unbounded wedge with apex~$v$
    (hatched in red in \cref{fig:rotation}) that goes
    from~$\overline{e_\ell}$ to~$\overline{e_\ell}(\alpha^*)$ in
    clockwise direction contains no vertices.
\end{enumerate}
We place the vertices $\ell_1,\ell_2,\dots,\ell_i$
between~$\overline{e_\ell}$ and~$x(\alpha^*)$.
Then property~(\ref{enum:propA'}') guarantees property~(\ref{enum:propA}).
Similarly, property~(\ref{enum:propB'}') and invariants~(I\ref{enum:active})
and~(I\ref{enum:outerface}) for~$\Gamma$ imply property~(\ref{enum:propB}).
The vertices $r_1,r_2,\dots,r_j$ are placed symmetrically
by rotating~$\overline{e_r}$ around~$v$ counterclockwise.

\subparagraph{Correctness.}
We now show that the invariants are maintained during the construction. 
By invariant~(I\ref{enum:active}), for each old vertex~$v$, the
region~$R_v$ is completely contained in the outer face of~$\Gamma$.
Hence, it is easy to define circular arcs for the new vertices
that satisfy invariant~(I\ref{enum:vertex}).
To show that invariant~(I\ref{enum:vertex}) also holds for the
circular arcs of the old vertices, we argue as follows.
By property~(\ref{enum:propA}) of the construction,
the parent edge~$e$ of a new vertex~$v$ can be decomposed as follows:
a line segment~$e_1$ that lies in~$R_v$ and a line segment~$e_2$ 
that lies in the triangle $\triangle \ell r v$ formed by the endpoints
of the parent edges of~$v$ (hatched in gray in \cref{fig:rotation}).

By invariant~(I\ref{enum:active}) for~$\Gamma$, the region~$R_v$ is empty
and, hence, $e_1$ does not intersect the circular arc of any old vertex.
By invariant~(I\ref{enum:vertex}) for~$\Gamma$, the circular arcs
of the old vertices lie in the outer face of~$\Gamma$
and, hence, it follows that $e_2$ also does not intersect the circular
arc of any old vertex.
Consequently, invariant~(I\ref{enum:vertex}) is maintained for the circular arcs of old vertices.

Invariant~(I\ref{enum:nonedge}) is retained for the non-edges that join two old
vertices since the circular arcs of these vertices have not been changed.
Property~(\ref{enum:propB}) and the fact that all new vertices are placed on $h$ imply 
that each of the non-edges incident to a new vertex~$w$ intersect~$C_w$.
Hence, invariant~(I\ref{enum:nonedge}) is also satisfied for the new non-edges.

Invariant~(I\ref{enum:active}) holds for the circular arcs of the
new vertices by invariant~(I\ref{enum:active}) for~$v$
in~$\Gamma$ and by~(I\ref{enum:vertex}) for the new vertices.
To see that invariant~(I\ref{enum:active}) holds for the circular
arcs of the old vertices, let~$u\neq v$ be an old vertex.
Let~$e$ be a parent edge of a new vertex and recall the definitions
of~$e_1$ and~$e_2$ from above.  The part~$e_1$ lies in~$R_v$ and~$e_2$
does not pass through the outer face of~$\Gamma$.
Hence, it follows that invariant~(I\ref{enum:active}) is retained for~$u$.

By invariant~(I\ref{enum:active}) for~$\Gamma$, 
the region~$R_v$ is contained in the outer face of~$\Gamma$.
Hence, by construction, invariant~(I\ref{enum:outerface}) holds for~$v$ and the new vertices.
Moreover, invariant~(I\ref{enum:outerface}) is also retained for the remaining vertices 
since, by construction, the edges incident to new vertices
intersect the outer face of~$\Gamma$ in~$R_v$ only.
\end{proof}

\section{Convex Outside Obstacle Representations}
\label{sec:conditions}

In this section we introduce a sufficient condition and a necessary
condition for a graph to admit a convex
OOR 
and then use these conditions to characterize the complements of trees
and cycles
that admit convex OORs.

We start with the sufficient condition.
Suppose that a graph~$G$ admits a convex OOR~$\Gamma$.
Let~$\sigma$ be the clockwise cyclic order of the
vertices of~$G$ along the convex hull of~$\Gamma$. 
If all neighbors of a vertex~$v$ of $G$ are consecutive in~$\sigma$,
we say that $v$ has the \emph{consecutive-neighbors property}, 
which implies that all non-edges incident to $v$ are consecutive around $v$ 
and trivially intersect the outer face 
in the immediate vicinity of~$v$; see \cref{fig:conditions:cnp}.
Note that this combinatorial condition is independent of the exact
location of the vertices as long as they are in convex position and
their clockwise order is fix.  This yields the following result.

\begin{lemma}[Consecutive-neighbors property]
  \label{clm:consecutiveNeighbors}
  Let $G$ be a graph, and let $\sigma$ be a cyclic order of
  $V(G)$.  If there is a subset $V'$ of $V(G)$ such that
  \begin{itemize} 
  \item every non-edge of $G$ is incident to a vertex in~$V'$ and
  \item every vertex in~$V'$ has the consecutive-neighbors property
    with respect to~$\sigma$,
  \end{itemize}
  then $G$ admits a convex OOR with cyclic vertex order~$\sigma$ on
  any set of points in convex position.
\end{lemma}

\begin{figure}[bh]
  \centering
  \begin{subfigure}[t]{.51\linewidth}
	\centering
 	\includegraphics[page=2]{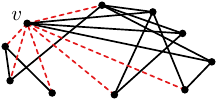}
	\caption{Vertex $v$ has the consecutive-neighbors property.}
	\label{fig:conditions:cnp}
  \end{subfigure}
  \hfill
  \begin{subfigure}[t]{.46\linewidth}
	\centering
 	\includegraphics[page=3]{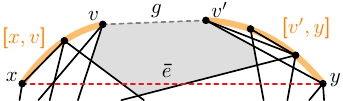}
	\caption{Gap~$g$ is a candidate gap for the non-edge~$\bar{e}$.}		
	\label{fig:conditions:candidategap}
  \end{subfigure}
  \caption{Examples for (\subref{fig:conditions:cnp}) the
    consecutive-neighbors property and
    (\subref{fig:conditions:candidategap}) a candidate gap.}
  \label{fig:conditions}
\end{figure}

If a graph~$G$ fulfills the condition stated in the above lemma, then
we say that $G$ has the \emph{consecutive-neighbors property.}
Note that the consecutive-neighbors property is not a necessary
condition for a graph to admit a convex OOR.  For example, we will
show that every grid graph admits a convex (even a regular) OOR
(\cref{thm:grids}), but the OORs that we construct do not fulfill the
consecutive-neighbors property (see, for example, vertex~$s$ in
\cref{fig:grid}).

Next, we derive the necessary condition.
For any two consecutive vertices $v$ and $v'$ in~$\sigma$
that are not adjacent in~$G$, we say that the line segment
$g = \overline{vv'}$ is a \emph{gap}.
Then the \emph{gap region} of $g$ is the inner face of $\Gamma+vv'$
incident to~$g$; see the gray region in \cref{fig:conditions:candidategap}.
We consider the gap region to be open, but add to it the relative interior
of the line segment $\overline{vv'}$, so that the non-edge $vv'$
intersects its own gap region.
Observe that each non-edge $\bar e = xy$ that intersects the outer face
has to intersect some gap region in an OOR.
Suppose that $g$ lies between $x$ and $y$ with
respect to~$\sigma$, that is, $[v,v'] \subseteq [x,y]$.
We say that $g$ is a \emph{candidate gap} for $\bar e$ if there is no
edge that connects a vertex in $[x,v]$ and a vertex in $[v',y]$.
(In \cref{fig:conditions:candidategap}, the two intervals are
highlighted in orange.)
Note that~$\bar e$ can intersect only gap regions of candidate gaps.

\begin{lemma}[Gap condition]
  \label{clm:gap-condition}
  A graph $G$ admits a convex OOR with cyclic vertex order~$\sigma$
  only if, for every non-edge of~$G$, there exists a candidate gap with
  respect to~$\sigma$.
\end{lemma}

It remains an open problem whether the gap condition is also sufficient.
Nonetheless, we can use the gap condition for no-certificates. 
To this end, we derived a SAT formula from the following expression,
which checks the gap condition for every non-edge of a graph~$G$:
\begin{equation*}
\bigwedge_{xy\notin E(G)} \!
	\left[\bigvee_{v\in [x,y)} \!\! 
		\left(\bigwedge_{u\in [x,v], w\in (v,y]} \!\!\! uw\notin E(G)\right) 
	\lor 
	\bigvee_{v\in [y,x)} \!\!
		\left(\bigwedge_{u\in [y,v], w\in (v,x]} \!\!\! uw\notin E(G)\right)
	\right]
\end{equation*}
We have used this formula to test whether all
connected cubic graphs with up to 16 vertices admit convex OORs.
The only counterexample that we found was the Petersen graph; see
\cref{fig:petersen}.
The so-called Blanu\v{s}a snarks, the Pappus graph, the dodecahedron, and
the generalized Peterson graph $G(11,2)$ satisfy the gap condition.
The latter three graphs do admit convex OORs~\cite{Gol21}.

\begin{figure}
  \null\qquad
  \begin{subfigure}[b]{.28\linewidth}
    \centering    
    \includegraphics[page=1]{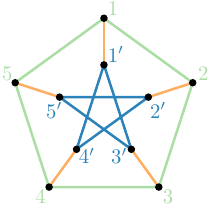}
    \caption{the Petersen graph}
  \end{subfigure}
  \hfill
  \begin{subfigure}[b]{.48\linewidth}
    \centering    
    \includegraphics[page=2]{petersen}
    \caption{a non-convex OOR of the Petersen graph~\cite{Gol21}}
  \end{subfigure}
  \qquad\null
  
  \caption{The Petersen graph does not admit a convex OOR.}
  \label{fig:petersen}
\end{figure}

The smallest graph (and the only 6-vertex graph) that does not
satisfy the gap condition is the wheel graph~$W_6$
(see \cref{clm:ngon:6vertices} in \cref{app:small-graphs}).
Hence, $W_6$ does not admit a {\em convex} OOR,
but it does admit a (non-convex) OOR; see~\cref{fig:not-convex}.

Next, we turn to dense graphs.

\subparagraph*{Complement of a Tree.}

For a graph~$G$, the graph~$\bar G$ with $V(\bar G)=V(G)$ and
$\bar E(G) = \set{uv \colon uv \not \in E(G)}$ is the \emph{complement}
of~$G$.  A \emph{caterpillar} is a tree that contains a path such that
all vertices are at distance at most~1 from the path.

\begin{theorem} \label{clm:caterpillar-tree}
  For any tree $T$, the graph~$\bar T$ has a convex OOR
  if and only if $T$ is a caterpillar.
\end{theorem}
\begin{proof}
We prove the statement in two steps.  
First, we show that, for every caterpillar~$C$, 
the graph~$\bar{C}$ has a convex OOR, in fact, a regular OOR.  
Then we show that, for every tree~$T$ that is not a caterpillar, 
$\bar T$ does not admit any convex OOR.

Let $C$ be a caterpillar, and let $\Pi=\croc{ p_1, p_2, \ldots, p_r }$
be a path in~$C$ such that every vertex in~$C$ has distance at most~1
from this path and such that $p_1$ and $p_r$ are vertices of degree~1
in~$C$; see \cref{fig:caterpillar} for an example.
For $i \in \set{2, \dots, r-1}$, let $\ell_1^i, \ell_2^i, \dots, \ell_{n_i}^i$ 
be the leaves adjacent to path vertex~$p_i$ (if any).
We arrange the vertices of~$\bar C$ in cyclic order as follows.  
First, we take the path vertices in the given order.
Then, for each $i \in \set{2, \dots, r-1}$, 
we insert the leaves adjacent to vertex~$p_i$ between~$p_i$ and~$p_{i+1}$ into the cyclic order; 
see \cref{fig:compl-caterpillar-oor}.

\begin{figure}[tb]
  \null\qquad
  \begin{subfigure}[b]{.35\linewidth}
    \centering
    \includegraphics[page=1]{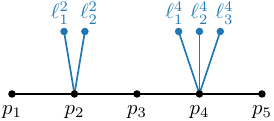}
    \caption{}
    \label{fig:caterpillar}
  \end{subfigure}
  \hfill
  \begin{subfigure}[b]{.35\linewidth}
    \centering
    \includegraphics[page=2]{caterpillar}
    \caption{}
    \label{fig:compl-caterpillar-oor}
  \end{subfigure}
  \qquad\null
  \caption{A caterpillar and a regular OOR of its complement.}
  \label{fig:compl-caterpillar}
\end{figure}   
  
The resulting cyclic order is $\sigma =
\langle p_1, p_2,  \ell_1^2, \ell_2^2, \ldots, \ell_{n_2}^2, \ldots, p_{r-1},
\ell_1^{r-1}, \ell_2^{r-1}, \ldots, \ell_{n_{r-1}}^{r-1}, p_r \rangle$.
Observe that every non-edge of $\bar{C}$ is incident to a vertex
in~$V(\Pi)$ and that, for every $i \in \{2,\dots,r\}$, vertex~$p_i$
in~$V(\Pi)$ has the consecutive-neighbors property with respect
to~$\sigma$ if we view its ``incoming'' non-edge from~$p_{i-1}$ as an
edge.  We may do this since this non-edge is still viewed as a
non-edge from its other endpoint.  Note that~$p_1$ trivially has the
consecutive-neighbors property.  Hence, by
\cref{clm:consecutiveNeighbors}, $\bar{C}$ admits a convex OOR with
cyclic vertex order~$\sigma$ on any set of points in convex
position, that is, $\bar{C}$ admits even a regular OOR.


Now we prove the second part of the statement.
Let $Y$ be the tree  that consists of a root~$c$ with three children~$\ell$, $m$, and $r$, 
each of which has one child, namely $\ell'$, $m'$, and $r'$, respectively; 
see \cref{fig:compl-tree-tree}.
Let $T$ be a tree that is not a caterpillar. 
Note that $T$ has a subtree that is isomorphic to~$Y$.

\begin{figure}[b]
  \centering
  \subcaptionbox{\label{fig:compl-tree-tree}}{\includegraphics[page=1]{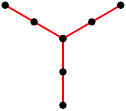}}
	\hfil
  \subcaptionbox{\label{fig:compl-tree-oor}}{\includegraphics[page=2]{K_7-tree}}
  \caption{(\subref{fig:compl-tree-tree}) The smallest tree $Y$ that is not a caterpillar and (\subref{fig:compl-tree-oor})
  a non-convex OOR of $\bar Y$.} 
  \label{fig:compl-tree}
\end{figure}

Let $\sigma$ be any cyclic order of $V(Y)$.    
We now show that $\bar{T}$ admits no convex OOR with respect to $\sigma$.
To this end, we find an edge $e$ of $Y$ (i.e., a non-edge of $\bar T$) 
that is a diagonal of a convex quadrilateral $Q$ formed by four
non-edges of~$Y$.  Observe that any non-edge of~$Y$ must
be an edge of~$\bar{T}$ (otherwise $T$ would contain a cycle).
Hence, the non-edge~$e$ of~$\bar{T}$ (being enclosed by a 4-cycle of
edges of~$\bar{T}$) does not have a candidate gap, which by
\cref{clm:gap-condition} implies that $\bar{T}$ does not admit a
convex OOR.

It remains to show the existence of~$e$ and~$Q$.
Without loss of generality, let $\croc{ c, r, m, \ell }$ be the order
of~$c$ and its children in~$\sigma$.  We distinguish four cases.
\begin{figure}[tb]
  \centering
  \begin{subfigure}[t]{.2\linewidth}
	\centering
	\includegraphics[page=1]{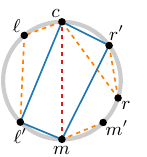}
	\caption{Case 1}
	\label{fig:tree-case1}
  \end{subfigure}
  \hfill
  \begin{subfigure}[t]{.2\linewidth}
	\centering
	\includegraphics[page=2]{cases-complement-tree}
	\caption{Case 2}
	\label{fig:tree-case2}
  \end{subfigure}
  \hfill
  \begin{subfigure}[t]{.2\linewidth}
	\centering
	\includegraphics[page=3]{cases-complement-tree}
	\caption{Case 3}
	\label{fig:tree-case3}
  \end{subfigure}
  \hfill
  \begin{subfigure}[t]{.2\linewidth}
	\centering
	\includegraphics[page=4]{cases-complement-tree}
	\caption{Case 4}
	\label{fig:tree-case4}
  \end{subfigure}
  \caption{Case distinction in the proof of \cref{clm:caterpillar-tree}:
	(\subref{fig:tree-case1})~Case~1: $c m$ is not intersected;
	(\subref{fig:tree-case2})~Case~2: $\ell \ell'$ intersects $c m$, $\alpha \cap Y \neq \emptyset$;
	(\subref{fig:tree-case3})~Case~3: $\ell \ell'$ intersects $c m$, $\alpha \cap Y = \emptyset$, 
	and $c$, $r$, $m'$ appear in this order in~$\sigma$;
	(\subref{fig:tree-case4})~Case~4: otherwise.}
  \label{fig:tree-cases}
\end{figure}
\begin{description}
\item[Case 1:] 
  None of the edges of $Y$ intersects $c m$; see \cref{fig:tree-case1}.

  \smallskip

  Then $e = c m$ lies inside the quadrilateral
  $Q = \langle c, r', m, \ell' \rangle$ formed by non-edges of~$Y$.
\end{description}
In the following three cases, we assume, without loss of generality,
that $\ell \ell'$ intersects $c m$.  Let $\alpha$ be the open
circular arc from $c$ to $\ell'$ in clockwise direction.
\begin{description}
\item[Case 2:] $\alpha \cap Y \neq \emptyset$, i.e., at least one
  vertex of $Y$ lies in $\alpha$, say $r$; see \cref{fig:tree-case2}.

  \smallskip

  Then $e = \ell \ell'$ lies inside the quadrilateral
  $Q = \langle \ell, r, \ell', m \rangle$.
\end{description}
In the remaining two cases, we assume that $\alpha \cap Y = \emptyset$.
\begin{description}
\item[Case 3:] The vertices $c$, $r$, and $m'$ appear in this order
  in~$\sigma$; see \cref{fig:tree-case3}.

  \smallskip

  Then $e = c r$ lies inside the quadrilateral
  $Q = \langle c, \ell', r, m' \rangle$.

  \medskip

\item[Case 4:] Otherwise; see \cref{fig:tree-case4}.

  \smallskip

  Then $e = m m'$ lies inside the quadrilateral
  $Q = \langle \ell, m', r, m \rangle$.
\end{description}
To conclude, $e$ and $Q$ always exist.
\end{proof}

\Cref{fig:compl-tree-tree} depicts the smallest tree~$Y$ that is not
a caterpillar and, hence, its complement~$\bar{Y}$ does not admit a
convex OOR.  The graph~$\bar{Y}$ does, however, admit an OOR;
see~\cref{fig:compl-tree-oor}.

\subparagraph*{Complete Graph Minus a Cycle.}

Using the gap condition (\cref{clm:gap-condition}), we can prove the
following theorem in a similar way as \cref{clm:caterpillar-tree}.
Let \emph{$C_k$} be a simple cycle of length~$k$.

\begin{theorem} \label{clm:complete-cycle}
  Let $3 \le k \le n$.  Then the graph~$G_{n,k} = K_n - E(C_k)$ admits a
  convex OOR if and only if $k \in \set{3, 4, n}$.
\end{theorem}
\begin{proof}
First, we show that, for $k \in \set{3, 4, n}$, 
the graph $G_{n,k}$ admits a convex OOR. 
To this end, we place the vertices $v_1, \dots, v_k$ of $C_k$ as an interval on a circle. 
If $k < n$, we place the remaining vertices in an arbitrary order, 
also as an interval, on the same circle.
For~$k = 3$, the vertex order of $C_3$ is determined; see \cref{fig:cycle-C3}.
For $k = 4$, we place the vertices of~$C_4$ in the order $\croc{ v_1, v_2, v_4, v_3 }$;
see \cref{fig:cycle-C4}. 
For $k = n$, we take the vertex order of~$C_n$; see \cref{fig:cycle-Cn}.
In the cases $k = 3$ and $k = n$, let $V' = V(G_{n,k})$.
In the case $k = 4$, let~$V' = V(G_{n,k}) \setminus \set{v_2, v_4}$ and
note that $v_2$	and $v_4$ are adjacent in $G_{n,k}$.
In all cases all vertices in $V'$ satisfy the consecutive-neighbors property 
and $V'$ covers all non-edges.
Therefore, by \cref{clm:consecutiveNeighbors}, the graph~$G_{n,k}$ admits a
convex OOR with	respect to the cyclic vertex order (depending on $k$) described above. 
Note that, in all cases, we fixed only the cyclic order of vertices 
and not their specific position. Thus, we can obtain the regular~OORs.
        
\begin{figure}[tb]
  \centering
  \begin{subfigure}[t]{.2\linewidth}
	\centering
	\includegraphics[page=1]{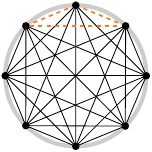}
	\caption{$G_{8,3}$}
	\label{fig:cycle-C3}
  \end{subfigure}
  \hfil
  \begin{subfigure}[t]{.2\linewidth}
	\centering
	\includegraphics[page=2]{complement-cycle-3-4-n}
	\caption{$G_{8,4}$}
	\label{fig:cycle-C4}
  \end{subfigure}
  \hfil
  \begin{subfigure}[t]{.2\linewidth}
	\centering
	\includegraphics[page=3]{complement-cycle-3-4-n}
	\caption{$G_{8,8}$}
	\label{fig:cycle-Cn}
  \end{subfigure}
  \caption{Regular OORs of the graph $G_{n,k}$ for $n=8$ and
		(\subref{fig:cycle-C3})~$k = 3$,
		(\subref{fig:cycle-C4})~$k = 4$,
		(\subref{fig:cycle-Cn})~$k = n$.}
  \label{fig:small-cycles}
\end{figure}    

\begin{figure}[b]
  \centering
  \begin{subfigure}[t]{.35\linewidth}
	\centering
	\includegraphics[page=1]{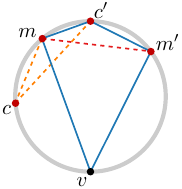}
	\caption{Case 1: $(c,c') \cap V(C_k) = \{m\}$}
	\label{fig:cycle-case1}
  \end{subfigure}
  \hfil
  \begin{subfigure}[t]{.35\linewidth}
	\centering
	\includegraphics[page=2]{cases-complement-cycle}
	\caption{Case 2: $\big|(c,c') \cap V(C_k)\big|>1$}
	\label{fig:cycle-case2}
  \end{subfigure}
  \caption{Case distinction in the proof of \cref{clm:complete-cycle}.}
  \label{fig:cycle-cases}
\end{figure}
	
Now let $k \in \set{5, \dots, n-1}$.
The graph~$G_{n,k}$ contains at least one vertex~$v$ 
that is adjacent to all other vertices.
Let $\sigma$ be any cyclic order of $V(G_{n,k})$ starting at~$v$ in clockwise direction.
We prove that $G_{n,k}$ does not admit a convex OOR with respect to~$\sigma$.
To this end, let $c$ be the first vertex of $C_k$ after~$v$ in~$\sigma$.
Let $c'$ be the last vertex in~$\sigma$ that is not adjacent to~$c$.
We consider two cases.
\begin{description}
\item[Case 1:] There is only one vertex~$m$ of~$C_k$ in the interval $(c,c')$;
see~\cref{fig:cycle-case1}.

  \smallskip

Note that $c m$ is a non-edge.  
Let $m'$ be the other vertex that shares a non-edge with~$m$.
Note that $m'$ lies between $c'$ and $v$ since $c$ is the first vertex of~$C_k$ after~$v$.
Hence, $mm'$ is a diagonal of the quadrilateral
$Q = \langle v, m, c', m' \rangle$.  
We argue that the edges of~$Q$ belong to~$G_{n,k}$.
By the choice of $v$, this is the case for $vm$ and $vm'$.
The same holds true also for $m c'$ and $c' m'$, 
otherwise the non-edges of~$G_{n,k}$ would contain a~$C_3$ or a~$C_4$, respectively.
However, $G_{n,k}$ has a simple $k$-cycle of non-edges with $k\ge5$.

  \medskip

\item[Case 2:] There are at least two
vertices of~$C_k$ in the interval $(c,c')$; see \cref{fig:cycle-case2}.

  \smallskip

Recall that $cc'$ is a non-edge of~$G_{n,k}$.
For $G_{n,k}$ to admit a convex OOR with respect to~$\sigma$,
by \cref{clm:gap-condition}, $cc'$ would have to have a
candidate gap~$g$, i.e., an edge $mm'$ of~$C_k$. Due to the presence of the
edges~$vc$ and~$vc'$, the gap~$g$ must lie on the side of
$cc'$ opposite of~$v$. By the definition of a candidate gap, no
edge connects the intervals $[c,m]$ and $[m',c']$.
This implies that $cm'$ and $mc'$ are non-edges.  
Hence, $\croc{ c, c', m, m' }$ is a 4-cycle of non-edges~-- a contradiction to the
fact that $G_{n,k}$ has a simple $k$-cycle of non-edges only with $k \ge 5$.
\end{description}
In both cases, we have that $G_{n,k}$ does not admit a
convex OOR with respect to~$\sigma$.
\end{proof}

\section{Regular Outside Obstacle Representations}
\label{sec:ngon}

This section deals with regular OORs of three graph classes.
A \emph{cactus} is a connected graph 
where every edge is contained in at most one simple cycle.
The \emph{weak dual} of a plane graph or a planar drawing is
its dual graph without the vertex corresponding to the outer face.
An \emph{outerpath} is a graph that admits an \emph{outerpath} drawing, i.e., 
an outerplanar drawing whose weak dual is a path.
Let $P_k = \langle v_1,\dots,v_k \rangle$ denote a (simple) path with $k$
vertices.  A graph $G$ is a \emph{grid graph} (or simply a grid)
if there are positive integers $k$ and $\ell$ such that
$G=P_k\square P_\ell$, that is, if $V(G)=V(P_k) \times V(P_\ell)$ and
if $G$ contains an edge between two vertices $(v_i,v_j)$ and
$(v_{i'},v_{j'})$ if and only if $|i-i'|=0$ and $|j-j'|=1$ or vice
versa.

\begin{theorem} \label{clm:ngon:cactus}
  Every cactus has a reducible regular OOR.
\end{theorem}
\begin{proof}
  For the given cactus~$G$, we first compute the block-cut tree (whose
  definition we recall below).
Then, following the structure of the block-cut tree, we treat the blocks one by one.
For each block, we insert its vertices as an interval into the vertex order of the subgraph that we have treated so far.
Finally, we prove that the resulting cyclic vertex order yields a reducible regular OOR.

A \emph{block-cut tree} of a connected graph is a tree that
has a vertex for each cut vertex and for each \emph{block}, which can either be a maximal biconnected subgraph or a bridge (cut edge).
There is an edge in the block-cut tree
for each pair $(B,v)$ of a block $B$ and a cut vertex $v$ with
$v \in V(B)$.  For an example of a block-cut tree of a cactus,
see~\cref{fig:cactus-block-cut-tree}.
We root the block-cut tree in an arbitrary block vertex and
number the block vertices according to a breadth-first search
traversal starting at the root.

\begin{figure}[htb]
  \centering
  \includegraphics[page=2]{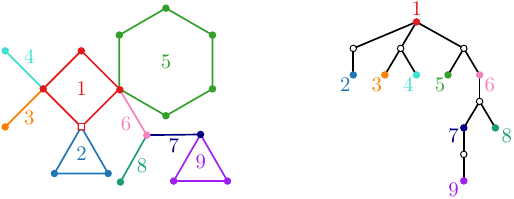}
  \caption{A cactus and its block-cut tree.
  White tree nodes correspond to cut vertices in the cactus.}
  \label{fig:cactus-block-cut-tree}
\end{figure}

In order to draw the cactus~$G$, we treat its blocks one by one and insert
the vertices of each block into the cyclic order, starting with the root
block. For each further block $B$, there is a cut vertex~$v_B \in V(B)
\cap V(B')$ where $B'$ is a block that we have treated earlier.
We insert the vertices of~$B$ as an interval between~$v_B$ and its
clockwise successor in the current cyclic vertex order.
For the root block~$B^\star$, let~$v_{B^\star}$ be an arbitrary vertex
of~$B^\star$ (marked by a white square in
\cref{fig:cactus-block-cut-tree,fig:cactus}).

Now we draw the current block $B$. If $B$ is a single edge $v_B w$,
we place $w$ immediately behind~$v_B$. If $B$ is a cycle with
$k$ vertices ($k \ge 3$), we start with the cut vertex $v_B$
and proceed in a zig-zag manner, mapping the vertices to positions
1 (which is $v_B$), $k, 2, k-1, \dots, \lceil(k+1)/2\rceil$; see~\cref{fig:cactus} (center).
For~$k \le 4$, all vertices of $B$ satisfy the consecutive-neighbors property.
For~$k \ge 5$, exactly two vertices do not satisfy this
property, but these two vertices are adjacent, 
so by \cref{clm:consecutiveNeighbors} all non-edges of $B$ intersect
the outer face as required.

\begin{figure}[bht]
  \centering
  \includegraphics[page=1]{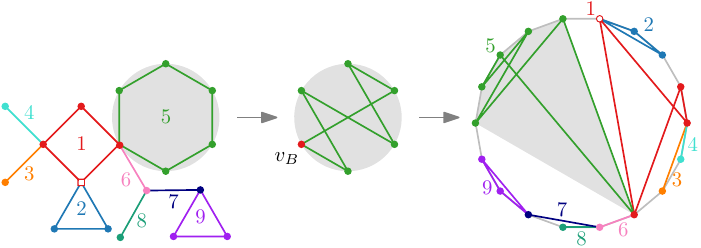}
  \caption{Constructing a reducible regular OOR of a cactus.}
  \label{fig:cactus}
\end{figure}
	
Now we draw $G$ by placing the vertices in the cyclic order on a circle
that we	just defined; the exact positions on the circle do not matter.
Consider the convex hulls of the blocks in the drawing.
Observe that any two of them share at most one
vertex. Moreover, the boundary of each convex hull lies
completely in the outer face of the drawing.
In the process described above, each block has its own OOR
due to \cref{clm:consecutiveNeighbors}.
Hence, the whole drawing is an OOR of~$G$.
The representation is reducible since each vertex has degree at most~2
within each block, and each block is surrounded by the outer face.

Since only the order of the vertices along the circle is important for
the OOR, not their exact positions, it is easy to obtain a regular
OOR.  For the same reason, even cactus forests admit OORs.
\end{proof}

\begin{theorem}
  \label{thm:grids}
  Every grid has a reducible regular OOR.
\end{theorem}

\begin{proof}
Let $k$ and $\ell$ be positive integers such that the grid is
$G = P_k \square P_\ell$.  We name the vertices of $P_k$
such that $P_k=\langle v_1,\dots, v_k\rangle$.
We assume that $k \ge \ell$.
We place the vertices of each copy of~$P_k$ in a zig-zag manner
on consecutive corners of a regular $k\ell$-gon, that is, in the order
$v_k, v_{k-2}, v_{k-4}, \dots, v_5, v_3, v_1, v_2, v_4, \dots, v_{k-1}$ if $k$ is odd
and in the order  $v_k, v_{k-2}, v_{k-4}, \dots, v_4, v_2, v_1, v_3, v_5, \dots, v_{k-1}$ if $k$ is even;
see \cref{fig:grid}.
If $\ell=1$, then $G$ is a path and the zig-zag order that we just
described defines a regular OOR, hence we may assume that $\ell \ge 2$ from now on.
The copies of $P_k$ (black in \cref{fig:grid}) are placed directly
one after the other.  This fixes our drawing of~$G$.

\begin{figure}[bht]
  \centering
  \includegraphics{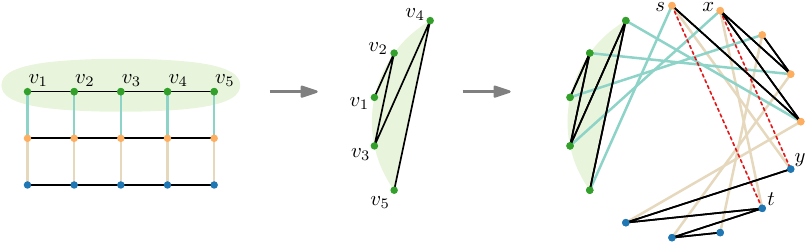}
  \caption{Constructing a reducible regular OOR of the grid $P_5 \square P_3$.}
  \label{fig:grid}
\end{figure}

Given a pair $\{s,t\}$ of vertices of~$G$, we say that the
\emph{cyclic length} of the edge or non-edge $st$ is~$d$ if the
line through~$s$ and~$t$ splits the plane into two open halfplanes
such that the one that contains fewer vertices of~$G$ contains exactly
$d-1$ vertices.
Within~$P_k$, the longest edge is $v_{k-1}v_k$; it has cyclic
length $k-1$.  Since the copies of $P_k$
are placed directly one after the other, every edge within a
copy of~$P_\ell$ (colored lightly in \cref{fig:grid}) has cyclic
length exactly~$k$.
	
We now show that every non-edge intersects the outer face. First,
consider a non-edge~$st$ that has cyclic length at least
$k+1$; see \cref{fig:grid}. Then it is longer than every edge
of~$G$; hence it intersects the gap region~$g$ between the first and the
last copy of~$P_k$. Note that non-edges of length exactly~$k$ exist
only between vertices of the first and the last copy of~$P_k$, but
these non-edges, too, intersect the gap region~$g$.

\begin{figure}[t]
  \centering
  \includegraphics{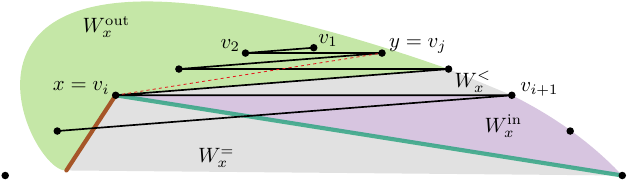}
  \caption{The four wedges with respect to a vertex~$x$ of the grid
  graph $P_k \square P_\ell$.  The black edges represent a copy
  of~$P_k$.  The two colored edges of length~$k$ belong to the
  same copy of~$P_\ell$.}
  \label{fig:wedges}
\end{figure}

\newcommand{\Win}{\ensuremath{W^\mathrm{in}}}
\newcommand{\Wout}{\ensuremath{W^\mathrm{out}}}

Next, consider a non-edge $xy$ that has length less than~$k$.
As every vertex of~$G$, $x$~is incident to one or two edges of
length~$k$ (contained in a copy of~$P_\ell$) and to one or two edges of length less than~$k$ (contained in a copy of~$P_k$).
Let $W^=_x$ be the wedge with apex~$x$ formed by (and including) the length-$k$ edge(s)~-- 
if there is just one such edge, then $W^=_x$ is the ray
starting in~$x$ that contains this edge; see \cref{fig:wedges}.
Similarly, let $W^<_x$ be the wedge formed by (and including) the shorter edge(s).
Hence, the angular space around~$x$ is subdivided into four wedges:
$W^=_x$, $W^<_x$, $\Win_x$, and $\Wout_x$, where $\Wout_x$ is the
(open) wedge between $W^=_x$ and $W^<_x$ that intersects the outer
face of the drawing in the vicinity of~$x$, and $\Win_x$ is the remainder
of the plane.  Let $W^=_y$, $W^<_y$, $\Win_y$, and $\Wout_y$ be defined analogously with respect to~$y$.
Note that the line segment $\overline{xy}$ is
neither contained in~$W^=_x$ nor in~$W^<_x$.
This is due to the fact that~$W^=_x$ contains only vertices of distance greater
than~$k$ and~$W^<_x$ contains no vertices in its interior. 
For the same reasons, $\overline{xy}$ is neither contained in~$W^=_y$ nor in~$W^<_y$.  
We now show that $\overline{xy}$ lies in~$\Wout_x$ or in~$\Wout_y$,
which implies that~$\overline{xy}$ intersects the outer face of the
drawing in the vicinity of~$x$ or~$y$, respectively.  We consider two cases.
\begin{description}
\item[Case~1:] Vertex $y$ belongs to the same copy of~$P_k$ as~$x$.

  \smallskip
  
  Let $i,j \in \{1,\dots,k\}$ be such that $x = v_i$ and $y = v_j$.
  First suppose that $j < i$; see \cref{fig:wedges}.  Then, by our
  layout, the cyclic length of~$xy$ is less than~$i$.  The
  non-edges in~$\Win_x$, however, have cyclic length at least~$i$. 
  (The edge $v_iv_{i+1}$ is the shortest edge in~$\Win_x$
  and has cyclic length~$i$; see \cref{fig:wedges}.)
  Hence,~$\overline{xy}$ cannot lie in~$\Win_x$ and must lie in~$\Wout_x$.  
  If $j > i$, we can argue analogously (by
  swapping~$x$ with~$y$ and $i$ with $j$) to show that~$\overline{xy}$
  lies in~$\Wout_y$.

  \medskip
  
\item[Case~2:] Vertex $y$ lies in a different (but neighboring) copy
  of~$P_k$, say, the next copy; see \cref{fig:grid} (right).

  \smallskip

  If~$\overline{xy}$ lies in~$\Wout_x$, we are done.  So suppose that
  $\overline{xy}$ lies in~$\Win_x$.  This implies that $x$ must lie in
  the first half (that is, $x \in \set{v_k, v_{k-2}, \dots, v_1}$)
  of its copy.  Since the cyclic length of~$xy$ is less
  than~$k$, $y$~must lie in the first half of its copy, too.  Due to
  our layout of~$P_k$, the (short) edges in~$W^<_y$ go to the other
  half of the copy.  Therefore, the length-$k$ edges that
  define~$W^=_y$ must lie between the short non-edge $\overline{xy}$ and
  the short edges that define~$W^<_y$. In other words, $\overline{xy}$
  lies in~$\Wout_y$.
\end{description}
  
For reducibility, we can argue similarly as for the non-edges.
Indeed, every edge of cyclic length~$k$ is incident to the gap region~$g$
between the first and last copy of~$P_k$; see \cref{fig:grid} (right).
The shorter edges
alternate in direction, so for $i \in \set{1, \dots, k-1}$, the edge
$v_iv_{i+1}$ of~$P_k$ is adjacent to the outer face in the vicinity
of vertex~$v_{i+1}$.
\end{proof}

\begin{theorem} \label{clm:ngon:outerpath}
  Every outerpath has a reducible regular OOR.
\end{theorem}
\begin{proof}
Let~$G$ be an $n$-vertex outerpath, and let~$\Gamma$ be an outerpath
drawing of~$G$. We show that~$G$ admits a reducible regular OOR.
The statement is trivial for $n \le 3$, so assume otherwise.
By reducibility and appropriately triangulating the inner faces
of~$\Gamma$, we may assume without loss of generality that each
inner face of~$\Gamma$ is a triangle. 
Let the path $\croc{ t_1, t_2, \dots, t_{n-2} }$ be the weak dual of~$\Gamma$.
Let~$V_i$ denote the set of vertices of~$G$ that are incident to the triangles $t_1, t_2, \dots, t_i$.
By definition, $V_1$ contains a vertex~$v_1$ of degree~2.
For $4 \le i \le n$, let $v_i$ denote the unique vertex
in $V_{i-2} \setminus \bigcup_{j=1}^{i-3}V_j$, where
$\bigcup_{j=1}^{i-3}V_j = \lbrace v_1,v_2,\dots,v_{i-1} \rbrace$;
see \cref{fig:outerpath} (left).  
For $4 \le i < n$, the vertex $v_i$ is incident to an inner
edge~$e_i = v_iv_j$ of~$\Gamma$ such that~$j < i$
and~$v_j$ belongs to the triangle~$t_{i-3}$.
Let $G_i = G[v_1,v_2,\dots,v_i]$.
We iteratively construct reducible regular OORs $\Gamma_3, \Gamma_4, \dots, \Gamma_{n}$
of $G_3, G_4, \dots, G_{n}(=G)$, respectively.
We create~$\Gamma_3$ by arbitrarily drawing~$G_3$ on the circle.
For $4 \le i < n$, to obtain~$\Gamma_i$ for from~$\Gamma_{i-1}$,
we consider the inner edge~$e_i = v_iv_j$ 
and place $v_i$ next to~$v_j$ on the circle,
avoiding the (empty) arc of the circle that corresponds to~$e_{i-1}$; see \cref{fig:outerpath} (right).
So if~$v_j$ is incident to multiple inner edges,
then the corresponding neighbors of~$v_j$ alternate on the circle.
The vertex~$v_n$ is placed next to~$v_{n-1}$, avoiding the arc that corresponds to~$e_{n-1}$.
(This yields that~$v_n$ has the consecutive neighbors property.)
Note that we fixed only the cyclic order of vertices and not their specific position.
Thus, we can place the vertices on a regular polygon.
  
\begin{figure}[htb]
  \centering
  \includegraphics{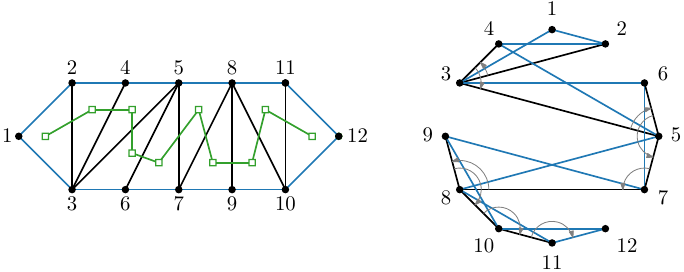}
  \caption{A drawing~$\Gamma$ of an outerpath~$G$ and a reducible
  regular OOR of~$G$ based on~$\Gamma$. Inner edges are black,
  outer edges are blue, weak dual edges are green.}
  \label{fig:outerpath}
\end{figure}
  
Now we show that each $\Gamma_i$, for $i \in \set{3, \dots, n}$,
is indeed a reducible regular OOR.
For any three points $a$, $b$, and $c$ on the circle~$C$,
let~$h_{ab}^{+c}$ be the open half-plane that is defined by the
line~$\ell_{ab}$ through $a$ and $b$ and that contains~$c$.
Similarly, let~$h_{ab}^{-c}$ be the open half-plane defined
by~$\ell_{ab}$ that does not contain~$c$.
Hence, $h_{ab}^{+c} \cup \ell_{ab} \cup h_{ab}^{-c} = \mathbb{R}^2$.

Let $u$, $v$, and $v'$ be the vertices of the triangle $t_{i-1}$
with $i \in \set{2, \dots, n-2}$ such that $vv'$ is an inner edge.
Furthermore, let $w$ be a new vertex of $t_i$ and, without loss
of generality, let $vw$ be an inner edge if $i \neq n-2$, see \cref{fig:outerpath-proof}.
By construction of the cyclic order described above,
we keep the invariant that when we place~$w$ on~$C$, $h_{vw}^{-u}$
is empty and $h_{v'w}^{-u}$ contains only~$v$ (among the vertices placed so far).

\begin{figure}[b]
  \centering
  \includegraphics{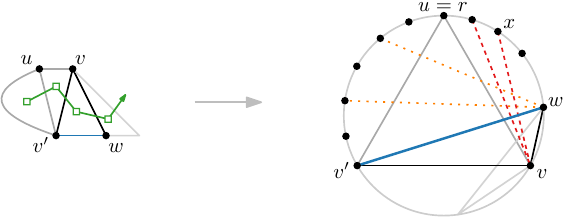}
  \caption{Constructing a representation for outerpaths such that the
  invariants are maintained.} 
  \label{fig:outerpath-proof}
\end{figure}

Now we show that when we add the vertex $w$, and, 
thus, the triangle $\triangle wv'v$,
a non-edge that goes through the outer face of the
drawing~$\Gamma_{i-1}$ of~$G_{i-1}$
continues to do so in the drawing~$\Gamma_i$  of~$G_i$.
We assume that $\triangle uv'v$ is oriented
counterclockwise (as in \cref{fig:outerpath-proof}).
Let~$r$ be the last neighbor of~$v$ in~$G_{i-1}$ in the
cyclic order that starts from $v$ and follows the circle counterclockwise.
Since $u$ is a neighbor of~$v$, $r = u$ (as in \cref{fig:outerpath-proof}) 
or $r$ lies strictly between $u$ and $v$, but $r \ne w$ because $w \not\in V(G_{i-1})$.
Note that the half-plane~$h_{v'w}^{-u}$ contains (the interior of) $\triangle wv'v$,
but $v$ is the only vertex in~$h_{v'w}^{-u}$.
Therefore, only non-edges {\em incident to~$v$} 
can be affected by the addition of $\triangle wv'v$, and
among these only the ones that go through the outer face
of~$\Gamma_{i-1}$ in the vicinity of~$v$. These are the non-edges
(dashed red in \cref{fig:outerpath-proof}) that are incident to~$v$
and lie in the half-plane~$h_{rv}^{-v'}$ that is induced by $rv$ and
does not contain~$v'$. Any such non-edge $vx$ intersects~$v'w$
since $v$ and $x$ lie on different sides of~$v'w$.
The intersection point of $vx$ and~$v'w$ lies on the outer face of~$\Gamma_i$
because~$h_{v'w}^{-u}$ contains only~$v$ and, in~$G_i$, 
$w$ is incident to only~$v$ and~$v'$.
This proves our claim regarding the ``old'' non-edges.

The ``new'' non-edges (dotted orange in \cref{fig:outerpath-proof})
are all incident to~$w$ and lie in~$h_{v'w}^{-v}$.
Since the two neighbors of~$w$, namely~$v$ and~$v'$,
are consecutive in~$\Gamma_i$, all non-edges incident 
to~$w$ go through the outer face~-- at least in the vicinity of $w$.

It remains to show that $\Gamma_i$ is reducible. For the two new
edges incident to~$w$ it is clear that they are both part of the
outer face~-- at least in the vicinity of~$w$.
Since~$h_{vw}^{-u}$ is empty, the only old edge that is affected
by the addition of~$\triangle wv'v$ is the edge~$vr$. It used to
be part of the outer face at least in the vicinity of~$v$.
Arguing similarly as we did above for the non-edge~$vx$,
we get that the intersection point of~$vr$ and~$v'w$ lies on the outer face.
\end{proof}

We consider two further simple graph classes that trivially admit regular OORs.
A graph is \emph{convex round} if its vertices can be cyclicly
enumerated such that the open neighborhood of every vertex is an
interval in the enumeration.
A \emph{bipartite graph} with bipartition $(U,W)$ of the vertex set is
\emph{convex} if $U$ can be enumerated such that, for each vertex
in~$W$, its neighborhood is an interval in the enumeration of~$U$.
Note that every complete bipartite graph is convex.
By definition, every convex round and convex bipartite graph admits a
cyclic order such that every vertex satisfies the
consecutive-neighbors property (\cref{clm:consecutiveNeighbors}).
This yields the following.

\begin{observation}
  Every convex round and convex bipartite graph admits a regular OOR.
\end{observation}

Our representations for convex round graphs, convex bipartite graphs,
cacti, outerpaths, and complements of caterpillars rely on the
consecutive-neighbors property (\cref{clm:consecutiveNeighbors}).
Hence, {\em any} convex point set of size $n$ is \emph{universal} in
the sense that it can be used as vertex set for an OOR for any graph
with $n$ vertices from one of these families.

\section{Small Graphs}
\label{app:small-graphs}

In this section we show that every graph with up to six vertices~--
except for the graph~$W_6$ depicted in \cref{fig:not-convex}~-- admits
a regular OOR (see \cref{clm:ngon:6vertices} below).
Lang~\cite{Lang22} showed that every outerplanar graph with up to
seven vertices admits a regular OOR.
The 8-vertex outerplanar graph in \cref{fig:outerplanar7}, however,
does not admit any regular OOR (see \cref{clm:ngon:outerplanar}
below).  It is the only 8-vertex outerplanar graph with this
property~\cite{Lang22}.

\begin{proposition} \label{clm:ngon:6vertices}
  There exists a regular OOR for every  graph with up to six vertices,
  except for the wheel graph $W_6$.
\end{proposition}

\begin{proof}
Note that $W_6$ is isomorphic to $G_{6,5} = K_6 - E(C_5)$.
Hence, by \cref{clm:complete-cycle}, $W_6$ does not admit a convex OOR.
Except for $G_{6,5}$, we claim that every graph with at most six
vertices satisfies the gap condition and admits a regular OOR.
For graphs with up to four vertices, this is not difficult to check.
For graphs with five vertices, see \cref{fig:reg5gon}.

\begin{figure}[tbh]
  \centering
  \includegraphics[width=\textwidth]{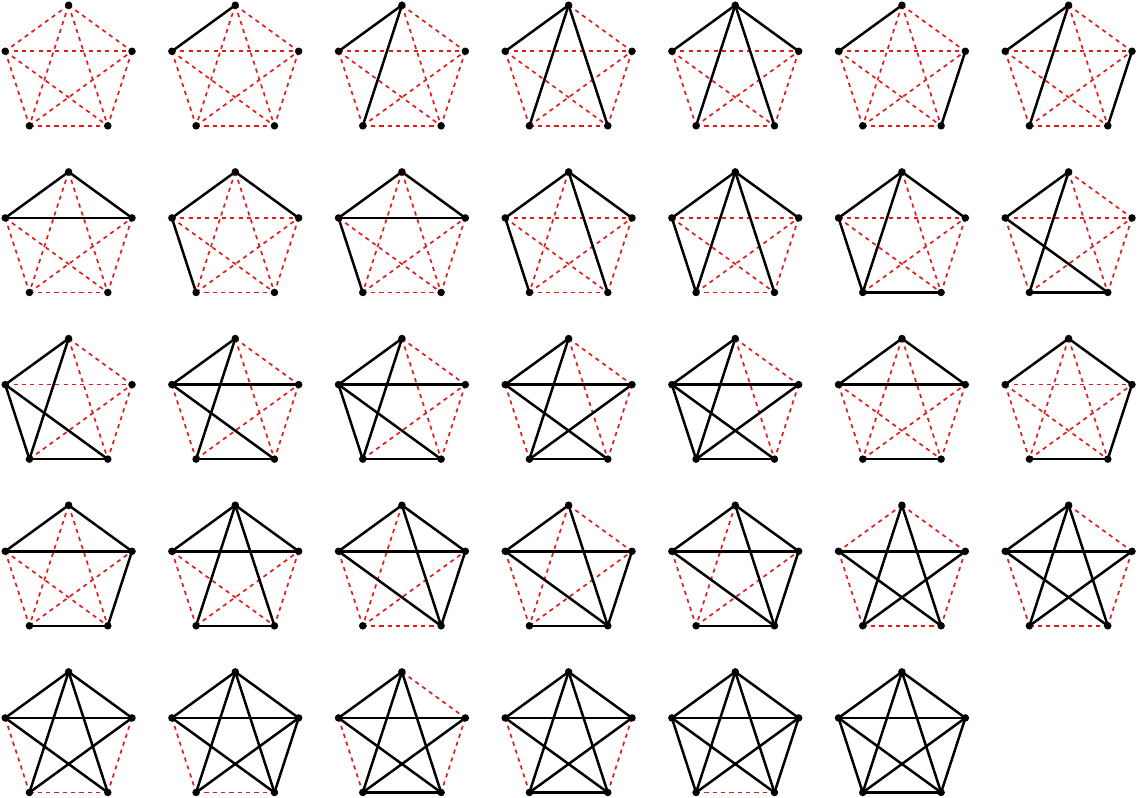}
  \caption{Every graph with up to five vertices admits a regular OOR.}
  \label{fig:reg5gon}
\end{figure}
	
For graphs with six vertices, we used the SAT formulation described in \cref{sec:conditions}
to find vertex orders that satisfy the gap condition for the vertex
set of the regular hexagon.
We now argue that each such vertex order either directly induces a
regular OOR or we can find an alternative OOR. 
To this end, we have to check whether every non-edge intersects a gap region.
We do a case distinction depending on the \emph{length} of the
non-edges in terms of the number of hexagon edges that they shortcut.
There are three cases for a non-edge~$e$ (orange dashed in
\cref{fig:reg6gon_1}).
\begin{itemize}
\item If $e$ has length~1, $e$ defines its own gap region, which it
  intersects.
\item If $e$ has length~2 (see \cref{fig:reg6gon_1-1,fig:reg6gon_1-2,fig:reg6gon_1-3}), then
  there are three possible gaps (up to symmetry, green dotted). For
  the gap condition to be satisfied for~$e$, in some cases, there must
  be further non-edges (purple dash-dotted).  In all three cases, $e$
  does intersect a gap region.
\item If $e$ has length~3 (see \cref{fig:reg6gon_1-4,fig:reg6gon_1-5}), there
  are two possible gaps (up to symmetry, green dotted). For the
  situation in \cref{fig:reg6gon_1-4}, the purple dash-dotted
  non-edge again makes sure that $e$ intersects the gap region of the
  green dotted non-edge. If, however, the endpoints of $e$ are not
  incident to the gap (see \cref{fig:reg6gon_1-5}), then $e$ may not
  intersect the gap region of the green dotted non-edge. This is the
  case if four specific vertex pairs are edges (blue solid). Then $e$
  touches the gap region, but does not intersect it. We analyze this
  case below.
\end{itemize}
	
\begin{figure}[tb]
  \centering
  \subcaptionbox{\label{fig:reg6gon_1-1}}{\includegraphics[page=1]{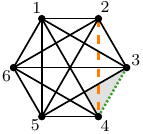}}
	\hfil
  \subcaptionbox{\label{fig:reg6gon_1-2}}{\includegraphics[page=2]{reg6gon-sep}}
	\hfil
  \subcaptionbox{\label{fig:reg6gon_1-3}}{\includegraphics[page=3]{reg6gon-sep}}
	\hfil
  \subcaptionbox{\label{fig:reg6gon_1-4}}{\includegraphics[page=4]{reg6gon-sep}}
	\hfil
  \subcaptionbox{\label{fig:reg6gon_1-5}}{\includegraphics[page=5]{reg6gon-sep}}
  \caption{Given an orange dashed non-edge~$e$, there are several
    (green dotted) candidate non-edges whose gaps could be intersected
    by~$e$.  Depending on the green dotted candidate non-edge, the
    existence of additional (purple dash-dotted) non-edges may be
    necessary to satisfy the gap condition.  Due to the blue solid
    edges, the drawing in~(\subref{fig:reg6gon_1-5}) is {\em not} a regular OOR.}
  \label{fig:reg6gon_1}
\end{figure}

\begin{figure}[tbh]
  \centering
  \includegraphics[page=2]{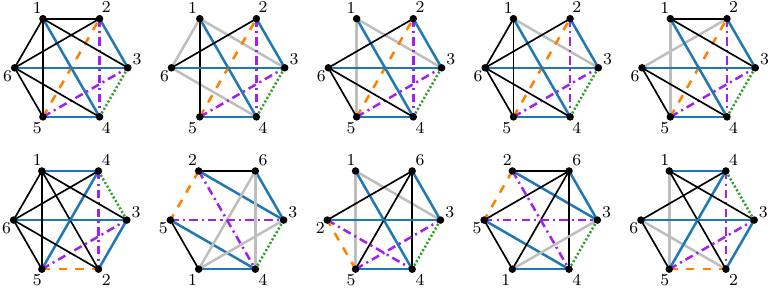}
  \caption{The first row depicts sub-cases of the last case from
    \cref{fig:reg6gon_1}; the second row depicts, for each of the upper
    sub-cases, an alternative vertex order that induces a regular OOR.}
  \label{fig:reg6gon_2}
\end{figure}
	
To solve the remaining case (with blue solid edges $\set{1,4}$,
$\set{2,3}$, $\set{3,6}$, $\set{4,5}$; and orange dashed, green
dotted, and purple dash-dotted non-edges $\set{2,5}$, $\set{3,4}$,
$\set{2,4}$, $\set{3,5}$, see \cref{fig:reg6gon_1-5}), we argue that
for every graph with this pattern, there exists a vertex order that
does induce a regular OOR.  In \cref{fig:reg6gon_2}, the first
row depicts five cases that cover all possible graphs with
the problematic pattern given the original vertex order;
the second row depicts a good vertex order for the same graph
(where the endpoints of the edge $\set{2,5}$ are consecutive). 
Our case distinction considers every subset~$S$ of the vertex pairs
$\set{\set{1,2},\set{1,6},\set{5,6}}$ as non-edges.  In each case, the
edges (black solid) and the non-edges (non-solid) are enforced by the
given case; gray line segments connect vertex pairs that can be either
edges or non-edges.
\begin{description}
\item[Case~1:] $S = \emptyset$ (first column of
  \cref{fig:reg6gon_2}).
  
\item[Case~2:] $S=\set{\set{1, 2},\set{5, 6}}$ and
  $S=\set{\set{1,2},\set{1,6},\set{5,6}}$ (second column).
  
\item[Case~3:] $S=\set{\set{1,2},\set{1,6}}$ (third column)
  and $S=\set{\set{1,6},\set{5,6}}$ (symmetric).

\item[Case 4:] $S=\set{\set{1,2}}$ (forth column)
  and $S=\set{\set{5,6}}$ (symmetric).

\item[Case 5:] $S=\set{\set{1,6}}$ (fifth column).
\end{description}
This shows that every graph with the problematic pattern admits a
regular OOR.
\end{proof}

\begin{proposition} \label{clm:ngon:outerplanar}
  There exists an 8-vertex outerplanar graph that has no regular OOR.
\end{proposition}
\begin{proof}
Consider the 6-vertex outerplanar graph~$G$ in \cref{fig:outerplanar6}.
Up to rotation and mirroring, it has only two regular OORs, 
which we tested using a variant of the SAT formulation described in \cref{sec:conditions}.
We call the resulting OORs Type~1 if the brown edge 
passes through the center of the regular hexagon,
and Type~2 if the purple edge 
does, see \cref{fig:outerplanar6}.
Let~$H$ be a supergraph of~$G$ such that the two new vertices~$u$
and~$w$ are incident to 
the endpoints of the brown and the purple edge,
respectively; see \cref{fig:outerplanar7}.
None of the possibilities for adding~$u$
and $w$ into the cyclic order of the vertices of~$G$ in
\cref{fig:outerplanar6} yields a regular OOR
since in each case one of the non-edges incident to~$u$ or to~$w$
(red dashed in \cref{fig:outerplanar7-perm})
lies completely in the interior of the drawing.  The vertex orders in
\cref{fig:outerplanar7-perm} satisfy the gap condition, but do not
admit regular OORs.
\end{proof}
	
\begin{figure}[hbt]
  \centering
  \includegraphics[page=1]{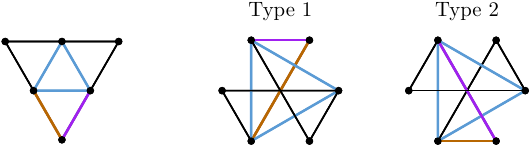}
  \caption{An outerplanar graph~$G$ and its regular OORs.}
  \label{fig:outerplanar6}
\end{figure}

\begin{figure}[hbt]
  \centering
  \includegraphics{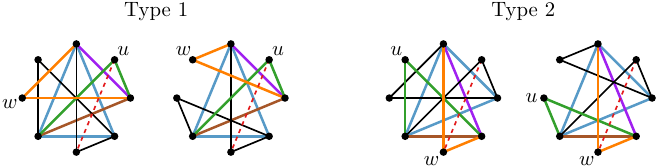}
  \caption{All possibilities for adding~$u$ and $w$ into the cyclic
    order of the vertices of~$G$ in \cref{fig:outerplanar6}.  In each
    drawing, the red dashed non-edge misses the outer face.}
  \label{fig:outerplanar7-perm}
\end{figure}

\section{Open Problems} 
\label{sec:open}

We leave the following problems open.
\begin{itemize}
  \item What is the complexity of deciding whether a given graph admits an OOR?
  \item Is the gap condition sufficient, i.e., does every graph with a
        cyclic vertex order satisfying the gap condition admit a convex OOR?
  \item Does every graph that admits a {\em convex} OOR also admit a {\em circular} OOR?
  \item Does every outerplanar graph admit a (reducible) convex OOR?
  \item Does every connected cubic graph {\em except the Peterson
      graph} (see \cref{fig:petersen}) admit a convex OOR?
  \item We have shown that every 2-tree and hence, every graph of
    treewidth~2, admits a reducible OOR.  Berman et
    al.~\cite{bcfghw-gong1-JGAA16} showed that there is a planar graph
    of treewidth~4 (see \cref{fig:x4}) that
    does not admit an OOR.  What about (planar) graphs of treewidth~3?
    The smallest 3-tree that does not admit a {\em convex} OOR is
    the (planar) graph with seven vertices shown in
    \cref{fig:3-tree}.  We verified this using the SAT formula
    described in \cref{sec:conditions}.
\end{itemize}

\begin{figure}[tb]
        \hspace{0ex}\hfill
	\begin{subfigure}[t]{.2\linewidth}
		\centering
		\includegraphics[page=1]{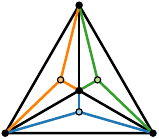}
		\caption{$B_7$}
		\label{fig:3-tree-graph}
	\end{subfigure}
	\hfill
	\begin{subfigure}[t]{.3\linewidth}
		\centering
		\includegraphics[page=2]{3-tree}
		\caption{a non-convex OOR of $B_7$}
		\label{fig:3-tree-OOR}
        \end{subfigure}
        \hfill\hspace{0ex}      
	\caption{The smallest 3-tree $B_7$ that does not admit a
          convex OOR.}
	\label{fig:3-tree}
\end{figure}

In particular, we conjecture the following:

\begin{conjecture}
  Every outerplanar graph admits a reducible circular OOR.
\end{conjecture}

The general idea is to show this only for one (infinite) family of
outerplanar graphs and to obtain the result for all outerplanar graphs
via reducibility.  In the family $(G_h)_{h \ge 1}$ that we
propose, the graph~$G_h$ is the \emph{complete
  outerplanar 2-tree} of height~$h$.  In the weak dual tree of this
graph, all dual vertices have degree either~$3$ or $1$, and one dual
vertex (the root) has distance $h$ to every leaf,
see~\cref{fig:outerplanar:graph}.

\begin{figure}[tb]
        \hspace{0ex}\hfill
	\begin{subfigure}[t]{.33\linewidth}
		\centering
		\includegraphics[page=1]{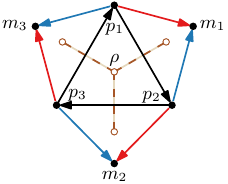}
		\caption{Graph $G_1$ with weak dual $T_1$}
		\label{fig:outerplanar:graphG1}
	\end{subfigure}
	\hfill
	\begin{subfigure}[t]{.38\linewidth}
		\centering
		\includegraphics[page=2]{outerplanar}
		\caption{Graph $G_3$ with weak dual $T_3$}
		\label{fig:outerplanar:graphG3}
	\end{subfigure}
        \hfill\hspace{0ex}
	\caption{Graph $G_h$ with weak dual $T_h$ (brown dashed edges)
          for $h \in \{1,3\}$.  In each vertex, the incoming edge from
          the primary (secondary) parent is drawn as a solid blue
          (red) arrow.
          Given the labels in~$G_1$, the new vertices in~$G_2$
          are labeled L$_i$ (R$_i$) if they are stacked to the left
          (right) of~$m_i$.
          For $h\ge 3$, the new vertices are labeled by
          appending~L or~R, respectively, to the label of the
          parent with the longer label.  (The subscript moves to the
          end of the label.)
          Hence, the label of a vertex~$v$ describes the path in~$T_h$
          from the root~$\rho$ to the face that was created by
          adding~$v$.}
	\label{fig:outerplanar:graph}
\end{figure}

For a parameter $0 < f \leq 1/2$, we then
construct a circular outside-obstacle representation of $G_h$ as follows.
Let $p_1$, $p_2$, and $p_3$ be the vertices of the triangle
corresponding to the root of the weak dual tree.
Place $p_1$, $p_2$, and $p_3$ with equal distances on the circle.
For~$j > 0$, place each level-$j$ vertex~$v$ clockwise
$(-f)^j \cdot 120^\circ$ away from its primary parent;
see \cref{fig:outerplanar:construction}.
So for odd $j$, $v$ is placed right before its primary parent, and for
even $j$, right after its primary parent (in clockwise order).

\begin{figure}[htb]
  \centering
  \includegraphics[page=3]{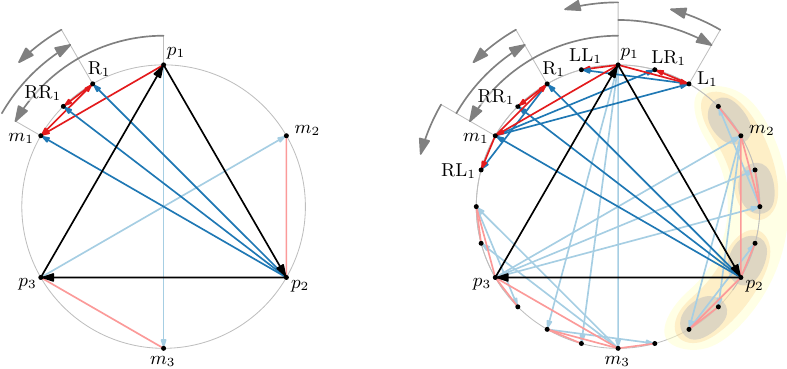}
  \caption{Constructing a circular outside-obstacle representation
    of~$G_3$ for $f = 1/2$.}
  \label{fig:outerplanar:construction}
\end{figure}

The intuition behind this construction is that every non-edge either
already goes through a gap region for $f = 1/2$, or moves towards one
when we decrease the value of~$f$ (which clusters the drawing).
One would then have to show that, for every~$h \ge 1$, there exists a
value $f_{\max}(h)$ such that the drawing defined above is an OOR (by
\cref{clm:gap-condition}).  Finding a recursive proof, however, turned
out to be challenging.  On the other hand, it is comparably easy to
show that the construction is reducible.

\pdfbookmark[1]{References}{References} 
\bibliographystyle{plainurl}
\bibliography{abbrv,obstacles}

\end{document}